\documentclass[letterpaper, 10 pt, journal]{ieeetran}  

\usepackage{graphics} 
\usepackage{epsfig} 
\usepackage{mathptmx} 
\usepackage{times} 
\usepackage{amsmath} 
\usepackage{amssymb}  
\usepackage{graphicx} 
\usepackage{stackrel}  
\usepackage{color}    %
\usepackage{epstopdf} 
\usepackage{cite}
\usepackage{mathrsfs}
\usepackage{latexsym}
\usepackage{amsthm}
\usepackage{amsfonts}
\usepackage{subfig}
\usepackage{lipsum}
\usepackage{cases}   
\usepackage{float}
\usepackage{arydshln}

\usepackage[pagebackref=false,colorlinks,linkcolor=blue,citecolor=magenta]{hyperref}  

\title{\LARGE \bf
  Convex Optimization of Bearing Formation Control of Rigid bodies on Lie Group
}

\author{Sara Mansourinasab$^{1}$, Mahdi Sojoodi$^{2}$ and S. Reza Moghadasi$^{3}$
	\thanks{$^{1}$Sara Mansourinasab is with the Department of Electrical and Computer Engineering, Tarbiat Modares University, Tehran, Iran 
		{\tt\small sara.mansouri@modares.ac.ir}}%
	\thanks{$^{2}$Mahdi Sojoodi is with the Department of Electrical and Computer Engineering, Tarbiat Modares University, Tehran, Iran
		{\tt\small sojoodi@modares.ac.ir}}%
	\thanks{$^{3}$S. Reza Moghadasi is with the Department of Mathematical Sciences, Sharif University of Technology, Tehran, Iran
		{\tt\small moghadasi@sharif.ir}}%
}

\begin{document}

\maketitle
\thispagestyle{empty}
\pagestyle{empty}
\noindent
\begin{abstract}
In this paper, the problem of reaching formation for a network of rigid agents over a special orthogonal group is investigated by considering bearing-only constraints as the desired formation. Each agent is able to gather the measurements with respect to other agents in its own body frame. So, the agents are coordinated-free concerning a global reference frame. Attracting to the desired formation is founded on solving an optimization problem for minimizing the difference between the instantaneous bearing between agents and their desired bearing. In order to have a unique global solution, the convex optimization method is implemented. Since the rotation matrices are not convex, the method of convex relaxation of rotation matrices space is used to embed the rotation matrices on the convex hull of the Lie group. Then the control law is designed to achieve the desired bearing with minimum energy consumption. Finally, a simulation example is provided to verify the results.
\end{abstract}
\section{Introduction}
The motion of a rigid body in space is a combination of rotational and translational movements. Group properties are preserved with the rigid body's motions. The set of positions of the rigid body is not a linear space, but a nontrivial smooth manifold. As a result, the rigid body’s orientation can be modeled by a unique rotational matrix, which gives a global and unique representation of its attitude. A detailed comparison of different methods of attitude representation is presented in \cite{hashim2019special}. 
In the case of multiple rigid bodies such as multi-satellite, multiple robotic explorers, and multiple spacecrafts, that bodies interact with each other as a multi-agent system, accurate measuring of mutual attitude plays an important part in consensus and formation of multi-agent systems. Different strategies of multi-agent systems have been extensively studied in many investigations \cite{wang2019coordinated,deng2020pose}. 

Formation control is an actively studied strategy in analyzing multi-agent systems \cite{liu2020robust,mansouri2022bearing}. In formation control, the type of available data that agents have access plays an important role in designing the control strategy. One of these data is the availability of the measuring distances between agents. In distance-based methods, the distances between agents are controlled to achieve the desired formation. In this method, the sensing capability is only defined with respect to local coordinate systems or body frame of agents and no GPS is required. Local coordinates formation control procedures include distance-based and bearing-based methods depending on the type of measuring and constraint parameters.

In bearing-based theory, the desired formation constraints are defined in terms of the direction of the neighboring agents. In other words, the direction of the unit vector aligned the edge that connecting two agents is considered as inter-agent bearing between those two agents. \cite{bishop2011stabilization,ern2012formation,zhao2014finite,sun2021vision,van2020pose,hoang2021robust,zhao2016bearingalmost,zhao2019bearing} have extensively studied different multi-agent strategies based on bearing rigidity theory for converging to the desired formation, generally assume each agent is able to measure bearings and distances of the nearest neighbors. Vision-based devices are suitable for bearing measurements. Optical cameras are one of low-cost and lightweight onboard sensors that are bearing-only sensors. In this paper, formation control for a multi-rigid body system is proposed considering bearing constraints.


Geometrically speaking, roto-translational motion of a rigid body represents a special Euclidean group $SE(3)$ state space configuration which is a Lie group. Using bearing measurements for analyzing such systems is compatible with geometry of the problem \cite{zelazo2014rigidity}. Considering the geometry and the nature of the $SE(3)$ manifold enables the rotation of the entire framework. This results in a synchronous rotation of all agents around their body attached framework with the same angular velocity that is coupled with the entire graph \cite{michieletto2016bearing}. As an illustrative example, multi-satellite systems that are orbiting around the Earth, should have a specific orientation to cover a specific region on Earth. As a result, in multi-rigid body problems, using bearing rigidity concept is beneficial, more accurate and low-cost.

Implementing optimization methods for solving attitude control problems of rigid bodies has been manifested in many applications. Satellite attitude estimation and control \cite{kim2004quadratic}, pose estimation \cite{horowitz2014convex}, and Mars rover pinpoint landing \cite{blackmore2010min,acikmese2007convex} are some of these applications. Motivated by the recent improvements in solving optimization problems, in this paper, minimization of the potential function is used for descending the distance between instantaneous and desired bearing. Since the optimization problem is introduced over the matrix spaces, different computational methods exists for optimization  over such spaces including linear matrix inequalities (LMI) \cite{choeung2023linear},  semidefinite programming (SDP) \cite{wang2020accurate}, positive semidefinite lift (PSD lift) \cite{matni2014convex} , and convex optimization \cite{rosen2015convex}. 
 Although rotational matrices are nonlinear; however, they are Lie group, which leads to the so called almost global convergence in consensus methods. 
On the other hand, the rotational matrices are non-convex, so the optimization problem on rotational matrices is a non-convex optimal problem. Approximating of the initial problem to a convex problem is one of the important methods for global non-convex optimization problem. This method is useful for global optimization using convexification in combinatorial optimal problems \cite{schrijver2003combinatorial}. Semidefinite relaxation over the convex hull of rotational matrices space $conv(SO(3))$ has been studied in various investigations. In \cite{saunderson2014semidefinite}, the $conv(SO(n))$ is represented as the projection of the intersection of PSD cone and an affine subspace with the trace equal to one (PSD lift). According to \cite{saunderson2014semidefinite}, the PSD lift is a spectrahedron, so the properties of a spectrahedron can be used in projecting and lifting process. The PSD lift properties that is presented in defining the $conv(sO(n))$ is considered by others to solve different optimal attitude control   problems. \cite{chen2017constrained} utilizes the optimization over the PSD lift in order to solve cone avoidance attitude control over $SO(n)$. Reaching consensus for rotational bodies over $SO(n)$ is studied in \cite{matni2014convex} exploiting a convex relaxation of the space of rotation matrices. Further investigations from this group presents a method for pose estimation based on spectrahedral representation of the orbitopes of $SE(n)$ in \cite{horowitz2014convex}. Their method is based on the convex relaxation of the non-convex rotational matrices space using LMI representation of the convex hull of $SE(n)$. \\
\, \, \, Inspiring by the methods of $SO(n)$ convexification  presented in the mentioned references, in this paper, we have studied the problem of reaching formation for a group of rigid-body agents with bearing constraints. Based on the desired bearing constraints, a sum of square objective function is defined to be minimized through a convex optimization method. Since both rotational and translational motions of agents is considered in this paper, the kinematics equation of motions makes the problem a non-convex problem. The convex relation method is used based on the replacing $SO(n)$ with its convex hull that is introduced in \cite{saunderson2014semidefinite}. Besides, as the equality  constraints are not convex, the dynamical equations have been disceretized using the $\exp$ map and then approximated. Finally, for the relaxed problem, the optimization control algorithm has been provided. The layout of the paper is as follows. Section II,  reviews some required mathematical preliminaries. Then the problem formulation is overviewed in this section. Section III,  relaxes the non-convex bearing formation optimization problem to a convex model by embeding the rotation matrices on the convex hull of $SO(3)$. Section IV, provides the formation control law algorithm. Discretizing and approximating are applied, so that  the convex optimization problem  is achieved, that can be solved using the results of section III. In section V, a simulation example is provided to verify the results of the paper. Section VI is dedicated to the conclusion.

\section{Mathematical Preliminaries and Problem Formulation}
\subsection{Notation}
The specific notations that are utilized in this paper are briefly summarize in this section. $\mathbb{S}^m$ is the space of $ m \times m$ real symmetric matrices and $\mathcal{S}^m_+ \subset \mathbb{S} $ depicts the  PSD matrices. $Tr(X)$ shows the trace of a square matrix $X$. For any matrices $X,Y \in \mathbb{R}^ {m \times n}$ the inner product denotes  $ \left< X,Y \right> = Tr(X^TY)$.

\subsection{Mathematical Preliminaries}
Consider a team of $n$ rigid bodies, as agents of the system with undirected communication graph $\mathcal{G}=(\mathcal{V},\mathcal{E})$, is located in the configuration space $SE(3)$, while  $\mathcal{V}=\{1,...,N\}$ is the set of vertices, $\mathcal{E}=\{(i,j) \in \mathcal{V} \times \mathcal{V} | j \in N_i\}$ represents the set of edges, and $N_i$ indicates the neighbors of agent $i$. 
As the system of rigid-body agents with rotational and translational motions is considered in this paper,  the set of rigid body's motions compounds the especial Euclidean group $SE(3)$.
For each agent $v_i \in \mathcal{V}$, the state $ R_i \in SO(3)$ is devoted such that the matrix group $ SO(3)=\{R \in \mathbb{R}^{3 \times 3} : R^TR=I_3, det(R)=1\}$ is the space of 3D rotations, where $I_3$ indicates the 3-dimensional identity matrix. The tangent space of $SO(3)$ in the identity element of the group, is the Lie algebra of $SO(3)$ written as \linebreak $\mathfrak{so}(3)=T_e SO(3)= \{ \Omega \in \mathbb{R}^ {3 \times 3} : \Omega^T = -\Omega \} $, which is the space of all $3 \times 3$  skew symmetric matrices. The tangent space of $SO(3)$ at any group member \linebreak $R$ is $ T_R SO(3)=\{RV: V \in \mathfrak{so}(3)\}$. 
There is a map between the tangent vector on the Lie group $ \Omega=\hat{\omega} \in \mathfrak{so}(3)$ with a vector in $\omega \in \mathbb{R}^3$ using the following $hat (\hat{.}) $ and $vee (\check{.}) $  maps

\begin{align}
	\small
	\omega=
	\left(
	\begin{array}{c}
		\omega_1 \cr\omega_2\cr\omega_3
	\end{array}
	\right)
	\in \mathbb{R}^3
	\stackrel[(\check{.})]{(\hat{.})}{\rightleftharpoons}
	\hat{\omega}=
	\left(
	\begin{array}{ccc}
		0&-\omega_3&\omega_2\cr\omega_3&0&-\omega_1\cr-\omega_2&\omega_1&0
	\end{array}
	\right)
	\in \mathfrak{so}(3). \label{eq25}
\end{align}
For small $\phi$, the $\exp$ map can be approximated  as
\begin{align}
	\exp[\alpha \phi] & =I+\alpha \phi \label{eq21}\\
	\exp[\alpha \phi] & =I+\alpha \phi + \frac{\alpha ^2}{2!} \phi^T \phi \label{eq22}
\end{align}
Although the space of rotational matrices is nonlinear, it is a Lie group. On the other hand, the set of $ n \times n$ rotational matrices is non convex. An important approach for global optimization of a non convex problem contains approximating of the original non convex problem to a convex optimization problem. Therefore, the convexification of non convex problems plays an important role in optimization problems. First of all, the problem is relaxed to the convex hull of $SO(3)$, which is shown by $conv(SO(3)$. Using the results of \cite{saunderson2015semidefinite,chen2017constrained}, the convex hull of $SO(3)$  is parameterized as spectrahedral projection of PSD lift. The $conv(SO(3))$ admits a $ 2^{(n-1)} \times 2^{(n-1)}$  spectrahedral representation. According to definition, a PSD lift of a convex set $ \mathcal{P} $ can be regarded as the linear projection of a spectrahedron.  In fact, the convex set $\mathcal{P}$ admits a PSD lift of size $d$ if it can be written as the linear projection of an affine slice of positive semidefinite cons $\mathcal{S}^d_+$ \cite{fawzi2017equivqrianr}.

PSD lift is defined as a set of PSD matrices with a unit trace as follows \cite{chen2017constrained}
\begin{align*}
	\mathcal{PL}=\left\{ Z \in \mathcal{S}^m_+ | \text{Tr}(Z)=1 \right\}.
\end{align*}
$\overline{\mathcal{PL}}$ which is the set of extreme points of $\mathcal{PL}$, is defined as
\begin{align*}
	\overline{\mathcal{PL}}=\left\{ Z \in \mathcal{S}^m_+ | \text{Tr}(Z)=1, \text{rank}(Z)=1 \right\}.
\end{align*}
Some results from PSD lift representation of $SO(3)$ are briefly summarized from \cite{saunderson2014semidefinite} as follows.

A matrix $\mathcal{A}$ has a spectrahedral representation such that the elements of it is a set of  $(A_{ij})_{1 \leq i,j \leq n}$ matrices with $A_{ij} \in \mathbb{S}^{2^{n-2}}$ and

\begin{align}
	A_{ij}=-P_e^T \lambda_i \rho_j P_e, \qquad 1 \leq i,j \leq n  \label{eq26}
\end{align}
while $P_{e} \in \mathbb{R}^{2^n \times 2^{n-1}}$ is the matrix with orthonormal columns 
\begin{align*}
	P_e=\frac{1}{2} 
	\left[ 
	\begin{array}{c}
		1\\
		1
	\end{array}
	\right]
	\otimes \overbrace{\sigma_0 \otimes \cdots \otimes \sigma_0}^{n-1}+\frac{1}{2}
	\left[ 
	\begin{array}{c}
		1\\
		-1
	\end{array}
	\right]
	\otimes \overbrace{\sigma_1 \otimes \cdots \otimes \sigma_1}^{n-1}, 
\end{align*}
and
\begin{align*}
	\sigma_0 &=
	\left[
	\begin{array}{cc}
		1 & 0 \\
		0 & 1
	\end{array}
	\right],
	\qquad
	\sigma_1=
	\left[
	\begin{array}{cc}
		1 & 0 \\
		0 & -1
	\end{array}
	\right], 
	\qquad
	\sigma_2=
	\left[
	\begin{array}{cc}
		0 & -1 \\
		1 &  0
	\end{array}
	\right],  \\
	\lambda_i &=\overbrace{\sigma_1 \otimes \cdots \otimes \sigma_1}^{i-1} \otimes \sigma_2 \otimes \overbrace{\sigma_0 \otimes \cdots \otimes \sigma_0}^{n-i}\\
	\rho_i &=\overbrace{\sigma_0 \otimes \cdots \otimes \sigma_0}^{i-1} \otimes \sigma_2 \otimes \overbrace{\sigma_1 \otimes \cdots \otimes \sigma_1}^{n-i}.
\end{align*}
The following theorem concludes that $conv(SO(3))$ can be written as a projection $\mathcal{A}$  of the PSD lift of size $2^{n-1} \times 2^{n-1}$. In fact, this theorem represents a convex hull definition for $SO(n)$. \\
\textbf{Theorem} \textbf{\cite{chen2017constrained}}:
The convex hull of $SO(n)$ is expressed as a projection of $2^{n-1} \times 2^{n-1}$ matrices with $trace=1$ as follows
\begingroup
\small
	\begin{align} 
		\text{conv}(SO(n))= \left\{ 
		\left[
		\begin{array}{ccc}
			\left< A_{11},Z \right> & \cdots & \left< A_{1n},Z \right>\\
			\vdots &  \ddots & \vdots\\
			\left< A_{n1},Z \right> & \cdots & \left< A_{nn},Z \right>\\
		\end{array}
		\right]
		\bigg| \ \  Z \geq 0, \ \ \text{tr}(Z)=1
		\right\}.  \label{27}
	\end{align}
\endgroup
In other words, the above theorem reveals that the matrix $X \in \text{conv}(SO(n))$ if and only if there exists a matrix $Z \in \mathbb{S}_+^{2^{n-1}}$ with $Tr(Z)=1$ such that
	\begin{align}\label{eq5-62} 
		X=\mathcal{A}(Z) \triangleq \left[
		\begin{array}{ccc}
			\left< A_{11},Z \right> & \cdots & \left< A_{1n},Z \right>\\
			\vdots &  \ddots & \vdots\\
			\left< A_{n1},Z \right> & \cdots & \left< A_{nn},Z \right>\\
		\end{array}
		\right],   
	\end{align}
and $A_{ij}$ is defined earlier. As a result,  $conv(SO(n))$ can be expressed as
\begin{align}
	\text{conv} (SO(n))= \left\{ \mathcal{A}(Z) \big| Z\geq 0, \text{tr}(Z)=1 \right\}. \label{eq28}
\end{align}
Besides, the map $ \mathcal{A} ^ \dagger : \mathbb{R} ^{n \times n} \to \mathcal{S}^{2^{n-1}}  $ is the adjoint  operator related to the inner product $ \left< \cdot , \cdot \right>$, and is calculated as
\begin{align}
	\mathcal{A}^\dagger (Y)= \sum_{i,j} A_{ij} Y_{ij}. \label{eq29}
\end{align}

\subsection{Problem Formulation}
 
 The following equations represents the kinematic equations of motions of the bodies
\begin{align}
\left\{ 
\begin{array}{ccc}
\dot{R}_i &= &R_i \hat{\omega}_i    \label{eq1} \\
\dot{p}_i &= &R_i v_i 
\end{array}  \right.
\end{align}
where $p_i=[p_{xi}, p_{yi}, p_{zi}]^T \in \mathbb{R}^3$ is the position of the $i$-th agent and $R_i \in SO(3)$ is the rotation matrix associated with the orientation of it. Furthermore, $v_i , \omega_i \in \mathbb{R}^3$ are respectively the linear and angular velocity of agent $i$ with respect to the body frame. For notation convenience, we use the position and attitude of the complete configuration in the stacked form as $\chi(i)=(\chi_p(i),\chi_r(i))=(p_i,R_i) \in SE(3)$.

The relative bearing measurement between agents $i$ and $j$ associated to the edge $e_k=(v_i,v_j)$ is defined as  \\
\begin{align}
b_k=b_{ij}=R_i^T \frac{p_i-p_j}{\|p_i-p_j\|}=R_i^T\bar{p}_{ij}	\label{eq2}
\end{align}
and
\begin{align*}
\bar{p}_{ij}=d_{ij} p_{ij}, \ \ p_{ij}=p_i-p_j, \ \ d_{ij}=\frac{1}{\|p_i-p_j\|}.
\end{align*}
\textbf{Definition 1}. The bearing rigidity function for a framework  $(\mathcal{G}, \chi )$ of n-agent is expressed as the map
\begin{align}
b_{\mathcal{G}}: SE(3)^n \to  \mathbb{S}^{2m}   \nonumber \\
b_{\mathcal{G}}(\chi) = [ b_1^T \dots b_m^T]^T  \label{eq3} 
\end{align}
that can be written in the following compact form
\begin{align}
b_\mathcal{G}(\chi)=\text{diag}(d_{ij} R_i^T)\bar{E}^Tp .  \label{eq4}
\end{align}
The bearing rigidity matrix in $SE(3)$ is defined as the gradient of the rigidity function as $ \text{B}_\mathcal{G}(\chi)= \nabla_{\chi} b_\mathcal{G}(\chi) \in \mathbb{R}^{3m \times 6n} $ that can be expressed as
\begin{align}
\text{B}_\mathcal{G}(\chi) &= \left[ \frac{\partial b_{ij}}{\partial p_i} \quad \frac{ \partial b_{ij}}{\partial R_i} \right]  \nonumber\\
& =\left[- \text{diag}(d_{ij} R_i^TP(\bar{p}_{ij})) \bar{E}^T \quad -\text{diag}(R_i^T\widehat{\bar{p}_{ij}}) \bar{E}_o^T \right]   \label{eq5}
\end{align}
such that $\widehat{\bar{p}_{ij}}$ is defined as (\ref{eq25}).

$\text{Tr}(A)$ represents the trace of matrix $A$.

\section{Convex Relaxation Formation Algorithm in $SO(3)$}
In this section, the formation control problem in $SO(3)$ with no common reference frame is studied. The desired formation is proposed as bearing only constraints. A network of $n$ rigid-bodies as agents is considered with the kinematic introduced in (\ref{eq1}).
The potential function $\phi$ is defined to describe bearing constraints. For converging to the desired formation, the potential function is only introduced in terms
of the bearing only constraints as
\begin{align}
\phi(\chi)=\frac{1}{2} \| b_\mathcal{G}^*-b_\mathcal{G} \|^2  \label{eq6}  
\end{align}
while $b_\mathcal{G}^*$ is the desired bearing rigidity function. This function is the stack form of the following function
\begin{align}
	\phi=\frac{1}{2} \sum_{(i,j) \in \mathcal{E}} \|b_{ij}^*-b_{ij} \|^2. \label{eq7}
\end{align}
In order to converge to the desired formation $ b_{\mathcal{G}}^*$, the cost function should be minimum as much as possible. Solving an optimization problem over the potential function (\ref{eq7}) is considered as the formation control strategy. In order to reach this goal,  the following minimization problem is introduced
\begin{align}
	\min_{R_i} \quad & \phi \label{eq8}\\
	 \text{s.t.} \quad &R_i \in SO(3).   \nonumber
\end{align}
Substituting from (\ref{eq2}) yields
\begin{align}
\min_{R_i} \quad &\frac{1}{2} \sum_{(i,j) \in \mathcal{E}} \|b_{ij}^*-R_i^T \bar{p}_{ij} \|^2  \label{eq9}  \\
\text{s.t.} \quad & R_i \in SO(3).   \nonumber
\end{align}
It is straightforward to conclude
\begin{align}
	\|b_{ij}^*-R_i^T \bar{p}_{ij} \|^2 & = \left< b^*_{ij}-R_i^T \bar{p}_{ij},b^*_{ij}-R_i^T \bar{p}_{ij}\right> \nonumber \\
	& =\left<b_{ij}^*, b_{ij}^*\right> -2\left<b_{ij}^* , R_i^T \bar{p}_{ij}\right>-\left<R_i^T \bar{p}_{ij},R_i^T \bar{p}_{ij}\right> \nonumber\\
	&=\|b_{ij}^*\|^2-2\left<b_{ij}^*,R_{ij}^T \bar{p}_{ij}\right>+ \|R_i^T\bar{p}_{ij}\|^2 \nonumber \\
	&=\|b_{ij}^*\|^2-2\left<b_{ij}^*,R_{ij}^T \bar{p}_{ij}\right>+1  \label{eq10}
\end{align}
Since the first and last terms are scalars and the middle term is just a non-constant term, so the minimization problem can be implemented on this term as follows
\begin{align}
	\min_{R_i} \quad & \sum_{(i,j) \in \mathcal{E}} -\left<b_{ij}^*,R_{ij}^T \bar{p}_{ij}\right> \nonumber \\
	\text{s.t.} \quad & R_i \in SO(3).	\label{eq11}
\end{align}
The objective function in (\ref{eq11}) is in the semidefinite programming form. However, as the rotational matrices are non-convex, so the above optimization problem is a non-convex problem. As a result, the original problem is convexified to the convex optimization problem by substituting the optimization over the convex hull $conv(SO(3))$ rather than $SO(3)$. 
\begin{align}
\max_{R_i} \quad & \sum_{(i,j) \in \mathcal{E}} \left< b_{ij}^*,R_{ij}^T \bar{p}_{ij}\right> \nonumber \\
\text{s.t.} \quad & R_i \in conv (SO(3)) 	 \label{eq12}
\end{align}
while $conv(SO(3))$ is defined in (\ref{eq27}). It is worth mentioning that, the optimization is just considered on the rotational motion of the body, while the complete optimization control problem is investigated in  section IV. 
Inspiring by \cite{matni2014convex}, the following lemma can be concluded.
\newtheorem{lem}{Lemma}
\begin{lem}
	$R^* = \mathcal{A}(Z^* ) \in SO(3)$, $Z^*=\eta ^T \eta$ is an optimal solution for the following optimization problem,
	\begin{align}
		\max_{R_i \in conv(SO(3))} \left< b^*_{ij}, b_{ij} \right> \label{eq14}
	\end{align}
	 while $\eta$ is the orthonormal eigenvector corresponding to the largest eigenvalue of $ \mathcal{A}^ \dag (\bar{p}_{ij} {b^*_{ij}}^T) $.
\end{lem}
\begin{proof}
	The problem is reformulated as
	\begin{align}
		\max_{R_i \in conv(SO(3))} \left< \bar{p}_{ij} {b^{*}_{ij}}^T, R_i \right> \label{eq15}
	\end{align}
	since
	\begin{align*}
		\left< b^*_{ij}, b_{ij} \right> = \left< b^*_{ij}, R_i^T \bar{p}_{ij} \right> = \left< b^*_{ij} \bar{p}^T_{ij}, R_i^T \right> = \left< \bar{p}_{ij} {b^{*}_{ij}}^T, R_i \right>. 
 	\end{align*}
	According to (\ref{eq27}), the convex hull of $SO(3)$ is parameterized as the projection $\mathcal{A}$ of the PSD lift. Therefore, the optimization problem (\ref{eq14}) is replaced by 
	\begin{align}
		\max_{Z} & \quad \left< \mathcal{A}^\dag (\bar{p}_{ij} {b^*_{ij}}),Z \right> \nonumber \\
		\text{s.t.} & \quad 
		\left\{
		\begin{array}{rrr}
			Z_i  \ge & 0  \quad \\
			\text{Tr}(Z_i)  =  & 1  \quad  \\
		\end{array}
		\right. \label{eq16}
	\end{align}
	Since the mapping $\mathcal{A}$ is an affine isomorphism between $conv(SO(3))$ and PSD lift, so a unique solution in $conv(SO(3))$ is achieved from the PSD lift under this projection \cite{chen2017constrained}. As a result, Slater's condition is satisfied, and thereafter strong duality is preserved. As the strong duality is satisfied, any pair of primal and dual optimal points must satisfy the KKT conditions. According to \cite{matni2014convex}, $\mathcal{A} (Z^*)$ is an extreme point of $conv(SO(3))$ while $Z^* \triangleq \eta \eta^T $ with rank=1 and $\eta$ equals the orthonormal eigenvector corresponding to the largest eigenvalue $\upsilon$ of $\mathcal{A}^\dag (\bar{p}_{ij} {b^*_{ij}})$.
\end{proof}
\newtheorem{Corollary}{corollary}
\begin{Corollary}
	The optimal solution $Z^*=\eta \eta^T$ is a unique solution to the optimization problem (\ref{eq16}) if the largest eigenvalue of $\mathcal{A}^\dag (\bar{p}_{ij} {b^*_{ij}})$ has multiplicity equals 1. 
\end{Corollary}
\begin{proof}
	Because of the same procedure, the proof is omitted. For more information, please refer to \cite{matni2014convex}.
\end{proof}

\begin{figure*}[!ht]
	\begin{minipage}[c]{0.5\textwidth}
		\vspace*{0pt}
		\centering
		\includegraphics[width=1\textwidth]{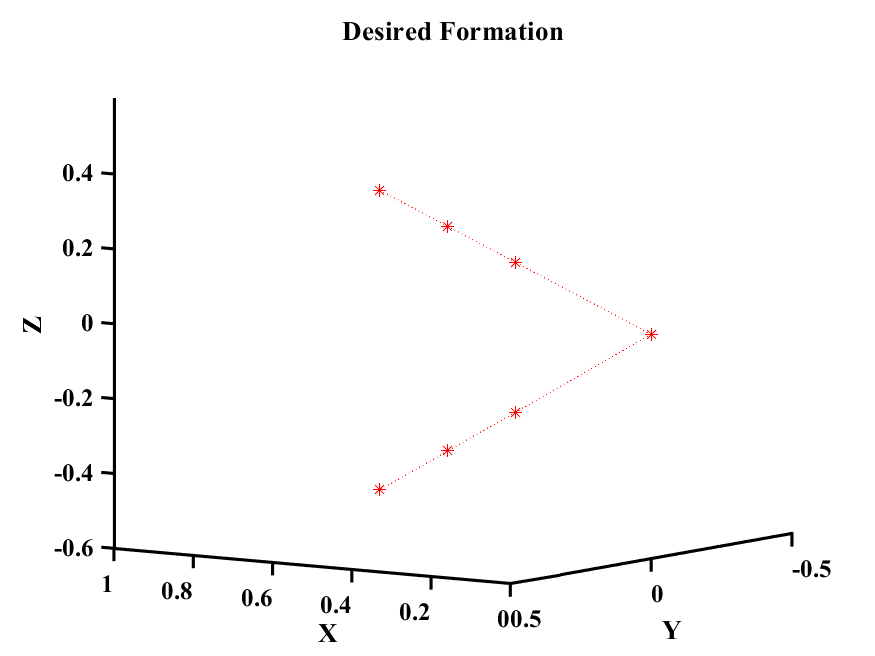}
		\centering
		\caption{The desired formation in a network of 7 agents}
		\label{fig1}
	\end{minipage}
	\begin{minipage}[c]{0.5\textwidth}
		\vspace*{0pt}
		\centering
		\includegraphics[width=1\textwidth]{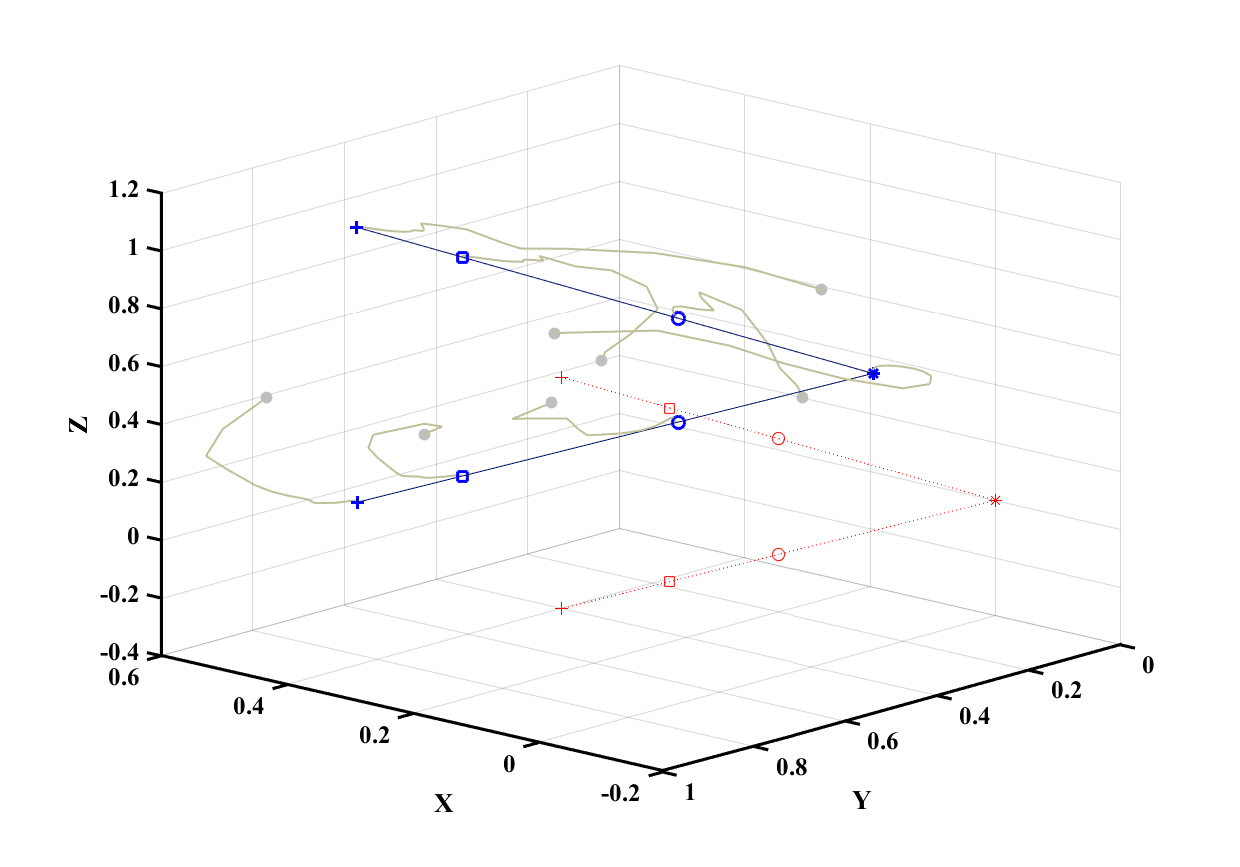}
		\centering
		\caption{Reaching the final formation using the convex optimization in a network of 7 rigid-body agents }
		\label{fig2}
	\end{minipage}
\end{figure*}

\section{Formation Control Law Design}
In previous section, the optimal solution $R^*$ is presented which minimizes the optimization problem (\ref{eq8}). In this section, the problem is redefined so that the formation problem  for multi rigid bodies system (\ref{eq1}) is achieved with minimum energy. In this case, the objective function (\ref{eq6}) changes into the following form
\begin{align}
	\phi &=\frac{1}{2} \|b^*_{\mathcal{G}}-b_{\mathcal{G}}\|^2 + \frac{1}{2} \left< W^1v,v \right> +\frac{1}{2} \left< W^2 \omega,\omega \right> \nonumber \\
	&= \frac{1}{2} \|b^*_{\mathcal{G}}-b_{\mathcal{G}}\|^2+\frac{1}{2} \left< Wu,u \right> \label{eq17}
\end{align}
while $W: \mathfrak{se}(3) \to \mathfrak{se}^*(3)$ is weighting function, $ W=\left( W^1 , W^2 \right) $ with symmetric positive definite matrices $W^1,W^2 \in \mathbb{R}^{3 \times 3}$, and  $u=(v, \omega)^T$.
Pursuing the same procedure as in section III, the optimization problem is specified as
\begin{align}
	{\underset{v,\omega}{min}}&\quad \sum_{(i,j) \in \mathcal{E}} \left< b^*_{ij}, b_{ij} \right>+ \frac{1}{2} \left< W_iu_i,u_i \right> \nonumber \\
	\text{s.t.}&\quad 
	\left\{
		\begin{array}{ccc}
		R_i&\in &conv(SO(3))   \\
		\dot{R_i}&=&R_i \hat{\omega}_i   \\
		\dot{p_i}&=&R_iv_i. 
		\end{array}
		\right. \label{eq18}
\end{align}

Before reforming the above problem to  the convex form, we firstly disceretize the dynamical systems' equations. For the attitude kinematics of each rigid body, $R_i$, the disceretized kinematic equations is as follows
\begin{align}
	R_i(k+1) &= \exp_{R_i(k)}(\epsilon R_i^{T}(k)\hat{\omega}_i(k)) . \label{eq20}
\end{align}
At each iteration $k$, the algorithm is solved using the current state $R_i(k)$ and moves along the geodesic opposite to the gradient so that the updated state $R_i(k+1)$ is obtained. In fact, the map $\exp_{R_i}(\cdot)=R_i \exp (\cdot)$ indicates the geodesic in Riemannian geometry. Using the $exp$ map is explained such that by moving from $R_i(k)$ along  geodesic, $R_i(k+1)$ is extracted. In fact, after solving the ODE equation, the calculated point is not a member of SO(3), so this point should be projected from the tangent space $\mathfrak{so}(3)$ to $SO(3)$ by $\exp_{R_i}$ map.  However, using $\exp_{R_i}$ means that moving forward from the initial point $R_i(k)$ in the amount of the geodesic which moves with $\epsilon\hat{\omega}_i$, the achieved point is located on $SO(3)$ and there is no need for projection. By choosing the proper step size $\epsilon$, the algorithm converges \cite{mansouri2022optimal}.

 As the equation (\ref{eq20}) is non-convex, the $\exp$ term is replaced by sequential linearization. As in \cite{jeon2009kinematic}, using (\ref{eq21}),(\ref{eq22}),  the first and second order of the linearized equation (\ref{eq20}) yields, respectively
 
 \begin{align}
 	R_i(k+1) &= R_i(I_3+ \epsilon R_i(k)^T \hat{\omega}_i(k))   \label{eq23} \\
 	R_i(k+1) &= R_i(I_3+ \epsilon R_i(k)^T \hat{\omega}_i(k) \nonumber \\
 	& \quad + \frac{\epsilon^2}{2}  ( R_i(k)^T \hat{\omega}_i(k) ) ^T ( R_i(k)^T \hat{\omega}_i(k) )) \nonumber \\
 	& =R_i(k)+\epsilon \hat{\omega}_i(k)- \frac{\epsilon^2}{2} R_i(k) \hat{\omega}_i(k) \hat{\omega}_i(k). \label{eq24}
 \end{align}
Since the equation (\ref{eq24}), is not in the SDP form of the second order linearized equation, this equation should be relaxed using Shur complement. It can  be noted that, due to similarity and  lack of space, the only first order approximation for $R_i(k+1)$ is considered in this paper. Second order approximation can be the topic of further investigations.

To sum up, the discretized kinematics equations of the system with first order attitude approximation are given by
\begin{align}
	R_i(k+1) &= R_i(k)+ \epsilon \hat{\omega}_i(k) \\
	P_i(k+1) &= R_i(k) v_i(k).  \label{eq27}
\end{align}
Therefore, the discretized form of the optimization problem (\ref{eq18}) is given by
\begin{align}
	{\underset{v,\omega}{min}}&\quad \sum_{(i,j) \in \mathcal{E}} \left< \bar{p}_{ij} {b^{*}_{ij}}^T, R_i \right>+\frac{1}{2} \left< W^1v_i,v_i \right> +\frac{1}{2} \left< W^2 \omega_i,\omega_i \right> \nonumber \\
	\text{s.t.}&\quad 
	\left\{
	\begin{array}{ccc}
		R_i \in   conv(SO(3)) \qquad  \quad \vspace{0.5mm}  \\
		R_i(k+1) = R_i(k)+ \epsilon \hat{\omega}_i(k)  \vspace{1.3mm}   \\
		P_i(k+1) =  R_i(k) v_i(k) \qquad  \vspace{1.3mm}\\
	\end{array}
	\right. \label{eq30}
\end{align}
The next corollary expresses the convex relaxation for the optimization problem (\ref{eq30}) using the results of section III.

Corollary: The convex relaxation of the above problem is achieved as follows
\begin{align}
\max_{v_i , \Omega_i} & \quad \sum_{(i,j) \in \mathcal{E}} \left< \mathcal{A}^\dag (\bar{p}_{ij} {b^*_{ij}}),Z_i \right> + \frac{1}{2} \left< \tilde{W}^1\xi_i,\xi_i \right> + \frac{1}{2} \left< \tilde{W}^2\Omega_i,\Omega_i \right> \nonumber \\
\text{s.t.} & \quad 
\left\{
\begin{array}{rrr}
Z_i  \ge & 0 \qquad \qquad \qquad \quad  \\
\text{Tr}(Z_i)  =  & 1 \qquad \qquad \quad \qquad  \\
Z_i(k+1) = & Z_i(k)+\epsilon \Omega_i(k) \qquad \\
X_i(k+1) = & X_i(k) + \epsilon \xi_i(k) Z_i(k)
\end{array}
\right. \label{eq19}
\end{align}
while $\Omega_i=\mathcal{A}^ \dagger(\hat{\omega}_i) $, $\xi_i(k)=\mathcal{A}^ \dagger(v_i)$ and $ X_i(k)= \mathcal{A}^ \dagger (p_i(k))$.

\begin{figure*}[!ht]
	\begin{minipage}[c]{0.5\textwidth}
		\vspace*{0pt}
		\centering
		\includegraphics[width=0.8\textwidth]{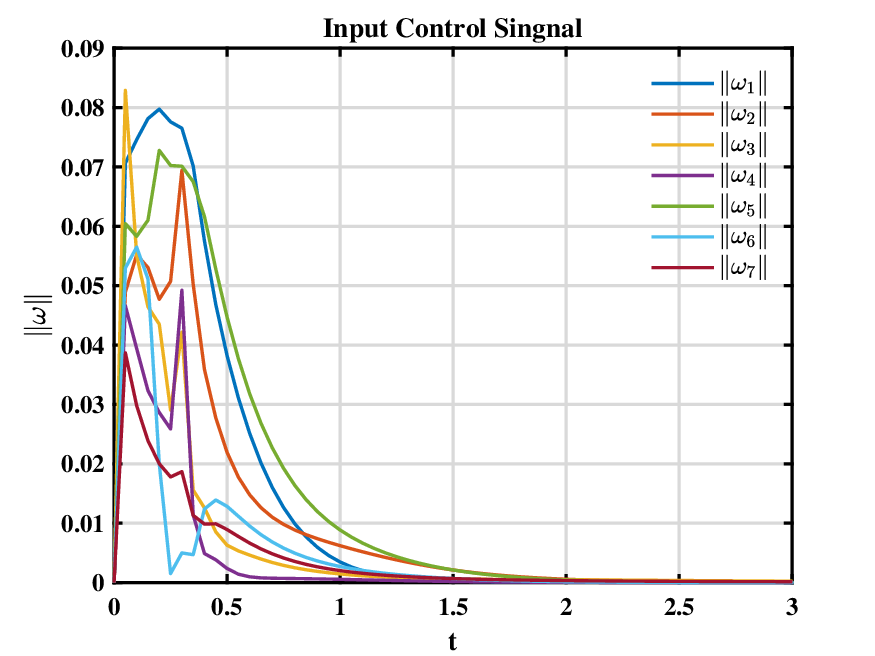}
		\centering
		\caption{The input control $\omega$ }
		\label{fig3}
	\end{minipage}
	\begin{minipage}[c]{0.5\textwidth}
		\vspace*{0pt}
		\centering
		\includegraphics[width=0.8\textwidth]{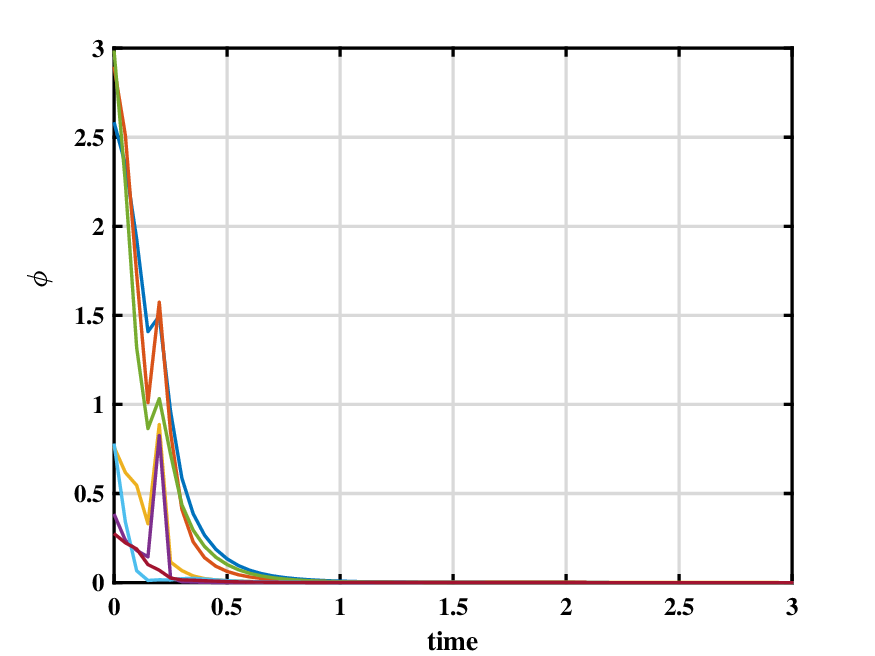}
		\centering
		\caption{The cost function $\phi_i$}
		\label{fig4}
	\end{minipage}
\end{figure*}

\section{Simulation Example}
In order to demonstrate the convergence of a network of rigid-body agents to a specified formation using the developed convex optimization method expressed in sections III and IV, in this section, an illustrative example is  conducted. The problem is constructed from seven agents with random initial orientations.  Converging the agents to the desired formation as the Fig. (\ref{fig1}) is the target of the optimiztion problem by minimizing the cost function (\ref{eq17}). The desired formation is defined as
$b^*=
\left[ 
b_{d_{12}}  b_{d_{23}}  b_{d_{34}}  b_{d_{15}}  b_{d_{56}}  b_{d_{67}} 
\right]$ 
matrix. Each column in This matrix relates to the bearing constraint for each edge.

 Initially, the optimization problem of (\ref{eq12}) that is relaxed to the convex problem (\ref{eq16}) is considered. The optimal $Z^*$ matrix for the 1-th agent is calculated as
 \begin{align}
 	Z^*_1= 
 	\left[
 	\begin{array}{cccc}
 		0.7011  &  0.1899 & -0.3764 &  0.1783 \\
 		0.1899  &  0.0515 & -0.1020  &  0.0483 \\
 		-0.3764 & -0.1020 &  0.2021  & -0.0957 \\
 		0.1783  &  0.0483 & -0.0957  &  0.0453
 	\end{array}
 	\right].  \nonumber
 \end{align}
The first agent's optimal attitude matrix, $R^*_1$ equals 
\begin{align}
	R^*_1= \mathcal{A}(Z^*)= \left[
	\begin{array}{lll}
		0.4930  &  0.6562  & -0.5713 \\
		-0.8494 &   0.5052 &  -0.1526 \\
		0.1885  &  0.5605  &  0.8064
	\end{array}
	\right]. \nonumber
\end{align}

Next, the formation control problem for achieving the optimal control signal for the cost function  (\ref{eq17})  is simulated. Similarly, the desired formation is given as (\ref{fig1}). As it is depicted in (\ref{fig2}), the agents converge to the desired formation. In this figure, the desired formation is shown red points, and the final formation is distinguished with blue points. The final formation of agents has the same relative bearing as the desired formation. Since the cost function is specifically defined in terms of bearing constraints, the final formation reaches the target formation with different scale. Besides, the control signal norm for all agents, and the cost function are depicted in figures (\ref{fig3}), and (\ref{fig4}), respectively. As it is conceived from the figures, the   desired relative bearing is achieved using minimum energy.

\section{Conclusion}
This paper has considered the formation problem for multi rigid-body agents on Lie group using optimization approach for minimizing the difference between instant  and desired bearings. In order to have a unique global response, the convex optimization method is implemented for minimizing the cost function. Since the rotational matrices space is not convex, the convex relaxation is done by embedding the rotational matrix space in convex hull  of the Lie group $SE(3)$. The applied method is based on the spectrahedral representation  of $conv{SO(3)}$.

The proposed  method in this paper presents a new approach for solving  formation strategies of multi agent systems with considering the geometry of the system's state space configuration. Since converging to the desired formation is based on the convex optimizing of the objective function, globally optimal convergence is guaranteed.

\end{document}